\documentclass[conference,letterpaper] {IEEEtran}
\usepackage[utf8]{inputenc} 
\usepackage[T1]{fontenc}
\usepackage{url}
\usepackage{ifthen}
\usepackage{cite}
\usepackage[cmex10]{amsmath} 
\ifx\du\undefined
  \newlength{\du}
\fi
\usepackage{siunitx}
\usepackage{tikz}
\usepackage{pgfplots}
\usepackage{verbatim}
\usepackage{color,soul}
\usetikzlibrary{spy,calc}
\usepackage{amsthm}
\usepackage{caption}
\newtheorem{theorem}{Theorem}
\newtheorem{lemma}{Lemma}

\usetikzlibrary{calc,positioning}
\usetikzlibrary{plotmarks}
\usepackage{multirow}
\usepackage{hhline}
\begin{document}
\title{Fundamental Limits of Coded Caching:\\ The Memory Rate Pair $\Big(K-1-\frac{1}{K},\frac{1}{(K-1)}\Big)$}
\author{\IEEEauthorblockN{Vijith~Kumar~K~P}
\IEEEauthorblockA{Department~of~EEE \\
IIT Guwahati, INDIA \\ 
Email: vijith@iitg.ac.in
} 
\and
\IEEEauthorblockN{Brijesh~Kumar~Rai}
\IEEEauthorblockA{Department~of~EEE \\
IIT Guwahati, INDIA \\ 
Email: bkrai@iitg.ac.in}
\and
\IEEEauthorblockN{Tony~Jacob}
\IEEEauthorblockA{Department~of~EEE \\
IIT Guwahati, INDIA\\
Email: tonyj@iitg.ac.in}
}
\maketitle
\begin{abstract}
Maddah-Ali and Niesen, in a seminal paper, introduced the notion of coded caching. The exact nature of the fundamental limits in this context has remained elusive even as several approximate characterizations have been found. A new optimal scheme for the $(3,3)$ cache network, operating at the memory rate pair $(5/3,1/2)$ for the demand where all the users request for distinct files, was introduced recently to partially  address this issue. In this paper, an extension of this scheme to the general $(K,K)$ cache network, operating at the memory rate pair $((K^2-K-1)/K,1/(K-1)$, is proposed. A new lower bound is also derived which demonstrates the optimality of the proposed scheme for the demand where all the users request for distinct files.

\end{abstract}
\section{Introduction}
The idea of coding in a cache network was introduced by Maddah-Ali and Niesen, in \cite{maddah2014fundamental}. They considered the $(N,K)$ canonical cache network, as shown in Fig. \ref{fig:canonical}. In this network, a server has $N$ files $\{W_{1},\dots,W_{N}\}$ which are of interest  to the $K$ users connected to it through a common shared link. The size of each file is $F$ bits and each user $U_{k}$ has an isolated cache $Z_{k}$ of size $MF$ bits, where $M\in[0,N]$. During the placement phase, the server fills the cache of each users with information regarding the $N$ files. In the delivery phase, each user communicates its request to the server. The users requests can be represented as a vector $\textbf{d}=\{W_{d_{1}},\dots, W_{d_{K}}\}$, where $W_{d_{k}}$ represents the file requested by user $k$. A set of packets $X_{\textbf{d}}$, consisting of $RF$ bits, is broadcast by the server over the common shared link to all users in response to the demand $\textbf{d}$. From the received packets $X_{\textbf{d}}$ and its cache contents, each user obtains its requested file. A memory rate pair $(M,R)$ is said to be achievable, if there exists a scheme that operates using $MF$ bits in each cache and $RF$ bits in the broadcast packets to satisfy each user's demand.

\begin{figure}[h]
\centering
\input{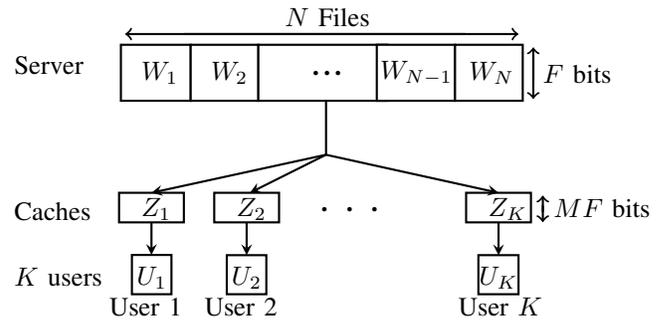}
\captionsetup{justification=centering}

\caption{The canonical cache network considered in \cite{maddah2014fundamental}.}
\label{fig:canonical}
\end{figure}

\begin{table*}[t]
\centering
\begin{tabular}{|c|c|c|c|c|}
\hline &&&&\\[-1em] Caching Scheme & Placement Phase & Cache Size, $M$ & Rate, $R$ & Matching Lower Bound  \\ 
\hline &&&&\\[-0.8em] Chen et al. \cite{chen2014fundamental}  & Coded pre-fetching  & $\dfrac{1}{K}$  & $N\bigg(1-\dfrac{1}{K}\bigg)$ & Cut set bound  \\[0.2em]
\hline &&&&\\[-0.8em] Amiri and Gunduz \cite{amiri2016fundamental}  & Coded pre-fetching  & $\dfrac{N-1}{K}$  & $N\bigg(1-\dfrac{N}{2K}\bigg)$ & No matching lower bound  \\[0.15em]
\hline &&&&\\[-0.7em] G{\'o}mez-Vilardeb{\'o} \cite{gomez2018fundamental} & Coded pre-fetching  & $\dfrac{N}{Kq}$, where $q\in\{1,2,\dots,N\}$  & $N\bigg(1-\dfrac{N+1}{K(q+1)}\bigg)$ &\\[-2.5em]&&&& Matching lower bound\\&&&& for $K=N$ and $M=1/(N-1)$ \\[0.8em]
\hline &&&&\\[-0.8em] Tian and Chen  \cite{tian2016caching}& Coded pre-fetching  & $\dfrac{t[(N-1)t+K-N]}{K(K-1)}$,&& \\&&where $t\in\{0,1,\dots,N\}$ &&\\[-2.3em]&&&$N\bigg(1-\dfrac{t}{K}\bigg)$ & No matching lower bound  \\[.8em]
\hline &&&&\\[-0.8em] Maddah-Ali and Niesen \cite{maddah2014fundamental} & Uncoded pre-fetching  & $\dfrac{Nr}{K}$, where $r\in\{0,1,\dots,K\}$  & $\bigg(\dfrac{K-r}{1+r}\bigg)$ &\\[-2.5em]&&&& Matching lower bound for\\ &&&&uncoded pre-fetching\\&&&& when $N\geq K$ \cite{wan2016optimality}  \\ 
\hline &&&&\\[-0.9em] Yu et al. \cite{yu2016exact} & Uncoded pre-fetching  & $\dfrac{Nr}{K}$, where $r\in\{0,1,\dots,K\}$  & $\dfrac{\binom{K}{r+1}-\binom{K-N}{r+1}}{\binom{K}{r}}$ &\\[-2.5em] &&&& Matching lower bound \\ [0.1em]&&&&for uncoded pre-fetching  \\ [0.15em]
\hline 
\end{tabular}
\caption{Previous works in coded caching.}
\label{Previous works}
\end{table*}

In this paper we consider the demands, where all the users request for distinct files. A coded caching scheme  for the $(N,K)$ cache network was proposed in \cite{maddah2014fundamental} for this demand. Lower bounds based on cutset arguments were also presented and proved to be within a multiplicative gap of 12 from the proposed scheme. Several improvements to these lower bounds were found in \cite{ ghasemi2017improved, sengupta2015improved, yu2017characterizing,wang2018improved}, with the current best multiplicative gap of 2, in \cite{yu2017characterizing}. Surprisingly, a modification of the scheme presented in \cite{maddah2014fundamental} was proved to be optimal when coding is permitted only in the delivery phase and the placement phase is restricted to be uncoded \cite{wan2016optimality,yu2016exact}. In the general case of coding being permitted in both  placement phase as well as delivery phase, several schemes were proposed \cite{chen2014fundamental,amiri2016fundamental,gomez2018fundamental,tian2016caching} to improve the performance in \cite{maddah2014fundamental}. These results are summarised in TABLE I. The exact rate memory tradeoff in this case has remained elusive except for the $(N,2)$ cache network characterized in \cite{tian2016symmetry}. As an improvement to this situation, in this paper, we present a new coded caching scheme operating at the memory rate pair $(\frac{K^2-K-1}{K},\frac{1}{K-1})$ which is shown to be optimal.

The rest of this paper is organized as follows. In section II, we consider the $(3,3)$ cache network as an example to explain the idea behind the proposed coded caching scheme. The new caching scheme is proposed in section III. In section IV, we demonstrate the optimality of the proposed scheme and conclude the paper in section V.

\section{The $(3,3)$ Cache Network} 
In this section we briefly review previous results that are relevant to the presentation that follows. In \cite{vijith}, the authors presented a new coding scheme for the $(3,3)$ cache network for the memory rate pair $(5/3,1/2)$. An extension to the $(4,4)$ cache network for the memory rate pair $(11/4,1/3)$ was also presented. For clarity, we present the $(3,3)$ cache network example, in detail, with updated notations. In this network the server has three files $\{W_{1},W_{2},W_{3}\}$, each of size $F$ bits, and three users are connected to the server through a common shared link.

In this scheme, the placement phase can be seen as executed in two stages. In the first stage, each file is split into 6 disjoint subfiles of size $F/6$ bits. Let the subfiles of the file  $W_{n}$ be represented as: 
\begin{IEEEeqnarray*}{rll}
\{W_{n,\{1,2\}}^{1},W_{n,\{1,2\}}^{2},W_{n,\{1,3\}}^{1},W_{n,\{1,3\}}^{3},W_{n,\{2,3\}}^{2},W_{n,\{2,3\}}^{3}\} 
\end{IEEEeqnarray*}
A collection of these subfiles are placed in the cache of each user as shown in TABLE II. In stage two, functions computed from these subfiles, as shown in TABLE II, are also placed in each user's cache. 

\begin{table*}[ht]
\hspace{0.05in}
\centering
\begin{tabular}{|c|c|c|c|c|c|c|}
\hline \multirow{7}{*}{Stage 1}&&\\[-1em] & User 1, $Z_{1}$&User 2, $Z_{2}$&User 3, $Z_{3}$ \\
\hhline{~---} & $W_{1,\{2,3\}}^{2}$ &$W_{1,\{1,3\}}^{1}$ &$W_{1,\{1,2\}}^{1}$\\
\hhline{~---}&  $W_{1,\{2,3\}}^{3}$ & $W_{1,\{1,3\}}^{3}$& $W_{1,\{1,2\}}^{2}$ \\
\hhline{~---}&  $W_{2,\{2,3\}}^{2}$ & $W_{2,\{1,3\}}^{1}$& $W_{2,\{1,2\}}^{1}$ \\
\hhline{~---} &  $W_{2,\{2,3\}}^{3}$ & $W_{2,\{1,3\}}^{3}$ & $W_{2,\{1,2\}}^{2}$\\
\hhline{~---} &  $W_{3,\{2,3\}}^{2}$ & $W_{3,\{1,3\}}^{1}$ & $W_{3,\{1,2\}}^{1}$ \\
\hhline{~---} &  $W_{3,\{2,3\}}^{3}$ & $W_{3,\{1,3\}}^{3}$ & $W_{3,\{1,2\}}^{2}$\\
\hhline{----} \multirow{4}{*}{Stage 2}& $W_{1,\{1,3\}}^{1}+ W_{1,\{1,2\}}^{1}$ & $W_{1,\{1,2\}}^{2}+ W_{1,\{2,3\}}^{2}$& $W_{1,\{2,3\}}^{3}+ W_{1,\{1,3\}}^{3}$\\
\hhline{~---} & $W_{2,\{1,3\}}^{1}+ W_{2,\{1,2\}}^{1}$ & $W_{2,\{1,2\}}^{2}+ W_{2,\{2,3\}}^{2}$ & $W_{2,\{2,3\}}^{3}+ W_{2,\{1,3\}}^{3}$\\
\hhline{~---} & $W_{3,\{1,3\}}^{1}+ W_{3,\{1,2\}}^{1} $& $W_{3,\{1,2\}}^{2}+ W_{3,\{2,3\}}^{2}$ & $W_{3,\{2,3\}}^{3}+ W_{3,\{1,3\}}^{3}$\\
\hhline{~---} & $W_{1,\{1,3\}}^{1}+ W_{2,\{1,3\}}^{1}+ W_{3,\{1,3\}}^{1}$ &$W_{1,\{1,2\}}^{2}+ W_{2,\{1,2\}}^{2}+ W_{3,\{1,2\}}^{2}$ &$W_{1,\{2,3\}}^{3}+ W_{2,\{2,3\}}^{3}+ W_{3,\{2,3\}}^{3}$\\
\hhline{----}
\end{tabular}
\captionsetup{justification=centering}
\caption{Cache contents for the $(3,3)$ cache network.}
\label{table:cache content}
\vspace{-2mm}
\end{table*}

In the delivery phase, the server computes the packets to be broadcast based on the users demand. Let us consider a demand $\textbf{d}=\{W_{d_{1}},W_{d_{2}},W_{d_{3}}\}$, where $W_{d_{1}}$ represents the file requested by the user $1$, $W_{d_{2}}$ represents the file requested by the user $2$ and $W_{d_{3}}$ represents the file requested by the user $3$, which all are assumed to be distinct. In response to this demand the server broadcasts a set of packets,
\begin{IEEEeqnarray*}{rll}
X_{\textbf{d}}=\left\{\begin{array}{c}
W_{d_{2},\{1,2\}}^{1}+W_{d_{3},\{1,3\}}^{1},\\ W_{d_{3},\{2,3\}}^{2}+W_{d_{1},\{1,2\}}^{2},\\W_{d_{1},\{1,3\}}^{3}+W_{d_{2},\{2,3\}}^{3}.
\end{array}
 \right\}
\end{IEEEeqnarray*}

Let us consider user 1 to understand how the scheme works. Based on its cache contents and broadcast packets, user 1 needs to compute all subfiles of $W_{d_{1}}$, namely,
\begin{IEEEeqnarray*}{c}
W_{d_{1},\{1,2\}}^{1},W_{d_{1},\{1,2\}}^{2},\\W_{d_{1},\{1,3\}}^{1},W_{d_{1},\{1,3\}}^{3},\\W_{d_{1},\{2,3\}}^{2},W_{d_{1},\{2,3\}}^{3}.
\end{IEEEeqnarray*}
The subfiles $W_{d_{1},\{2,3\}}^{2}$ and $W_{d_{1},\{2,3\}}^{3}$ are already present in its cache. It also has subfiles,
\begin{IEEEeqnarray*}{c}
W_{d_{2},\{2,3\}}^{3}\text{ and } W_{d_{3},\{2,3\}}^{2},
\end{IEEEeqnarray*}
which along with the broadcast packets,
\begin{IEEEeqnarray*}{c}
W_{d_{1},\{1,3\}}^{3}+W_{d_{2},\{2,3\}}^{3} \text{ and } W_{d_{1},\{1,2\}}^{2}+W_{d_{3},\{2,3\}}^{2}
\end{IEEEeqnarray*}
enables the user 1 to compute the subfiles $W_{d_{1},\{1,3\}}^{3}$ and $W_{d_{1},\{1,2\}}^{2}$. The cache content $W_{d_{2},\{1,3\}}^{1}+W_{d_{2},\{1,2\}}^{1}$ along with the broadcast packet $W_{d_{2},\{1,2\}}^{1}+W_{d_{3},\{1,3\}}^{1}$ can be used to compute the function,
\begin{equation}
\label{evaluate3,3}
W_{d_{2},\{1,3\}}^{1}+W_{d_{3},\{1,3\}}^{1},
\end{equation}
This function along with the cached content,
\begin{equation*}
W_{d_{1},\{1,3\}}^{1}+W_{d_{2},\{1,3\}}^{1}+W_{d_{3},\{1,3\}}^{1},
\end{equation*}
can be used to recover the subfile $W_{d_{1},\{1,3\}}$. This subfile along with the cached content,
\begin{equation*}
W_{d_{1},\{1,3\}}^{1}+W_{d_{1},\{1,2\}}^{1},
\end{equation*}
can recover the subfile $W_{d_{1},\{1,2\}}^{1}$. Thus all subfiles of the file $W_{d_{1}}$ can be successfully computed by user 1. 

Since all subfiles and computed functions stored in each users cache are of size $F/6$ bits and there are ten of them, the size of each cache is $(10/6)F=(5/3)F$ bits. In response to every possible demand, the server broadcasts three packets of size $F/6$ bits each and thus transmits $F/2$ bits on the common shared link. The proposed scheme can thus be seen to operate successfully at the memory rate pair $(5/3,1/2)$.
\begin{figure}[!t]
\centering
\begin{tikzpicture}[line cap=round,line join=round,x=2.6cm,y=1.7cm,
    spy/.style={%
        draw,green,
        line width=1pt,
        rectangle,inner sep=0pt,
    },
]

    \def\spyviewersize{3cm}

    \def\spyonclipreduce{.3pt}

    \def\spyfactorI{10}
    \coordinate (spy-on 1) at (5/3,0.52);
    \coordinate (spy-in 1) at (2.3,1.5);

    \def\pic{\coordinate (O) at (0,0);
       \draw [ultra thin,step=3,black] (0,0) grid (3,3);
%
   \foreach \x in {0,1,2,3}
   \draw[shift={(\x,0)},color=black,thin] (0pt,1pt) -- (0pt,-1pt)
                                   node[below] {\footnotesize $\x$};
      \foreach \y in {0,1,2,3}
      \draw[shift={(0,\y)},color=black,thin] (1pt,0pt) -- (-1pt,0pt)
                                    node[left] {\footnotesize $\y$};
  \draw[color=black] (5cm,-18pt) node[left] { Cache size $M$};
  \draw[color=black] (-0.5cm,3.5cm) node[left,rotate=90] { Rate $R$};
  \draw[smooth,gray!80!white,samples=1000,domain=0.0:2.2,mark=otimes,line width=2pt] 
  
   {(0,3)--(1/3,2)node[mark size=3.4pt,line width=0.1pt]{$\pgfuseplotmark{*}$}--(2/3,4/3)node [mark size=3pt,line width=0.1pt,label={left:$(2/3,4/3)$}]{$\pgfuseplotmark{*}$}--(1,1)--(7/6,5/6)node [mark size=3pt,line width=0.1pt,label={left:$(7/6,5/6)$}]{$\pgfuseplotmark{*}$}--(5/3,1/2)node [mark size=3.4pt,line width=0.1pt]{$\pgfuseplotmark{*}$}--(2,1/3)node [mark size=3.4pt,line width=0.1pt]{$\pgfuseplotmark{*}$}--(3,0)}; 
   \draw[smooth,blue!70,mark=otimes,samples=1000,domain=0.0:2.2,mark = $\otimes$,dash pattern=on 10pt off 10pt,line width=2pt]
      {(0,3)--(1/3,2)  node[mark size=2.5pt,line width=0.1pt,label={right:$(1/3,2)$}]{$\pgfuseplotmark{square*}$}--(1/2,5/3)  node[mark size=2.5pt,line width=0.1pt,label={right:$(1/2,5/3)$}]{$\pgfuseplotmark{square*}$}--(1,1) node[mark size=2.5pt,line width=0.1pt,label={right:$(1,1)$}]{$\pgfuseplotmark{square*}$}--(2,1/3) node[mark size=2.5pt,line width=0.1pt,label={right:$(2,1/3)$}]{$\pgfuseplotmark{square*}$}--(3,0)};
      \draw[,smooth,red!60!black,samples=1000,domain=0.0:2.2,mark size=2pt,mark =otimes*,dash pattern=on 6pt off 6pt,line width=2pt]
            {(0,3)--(1/3,2)--(1/2,5/3)--(1,1)--(5/3,1/2) node[mark size=4pt,line width=0.1pt,label={below:$(5/3,1/2)$}]{$\pgfuseplotmark{triangle*}$}--(2,1/3)--(3,0)};
    }

    \pic

    \node[spy,minimum size={2*\spyviewersize/\spyfactorI}] (spy-on node 1) at (spy-on 1) {};
    \node[spy,minimum size=\spyviewersize] (spy-in node 1) at (spy-in 1) {};
    \begin{scope}
        \clip (spy-in 1) circle (0.5*\spyviewersize);
        \pgfmathsetmacro\sI{1/\spyfactorI}
        \begin{scope}[
            shift={($\sI*(spy-in 1)-\sI*(spy-on 1)$)},
            scale around={\spyfactorI:(spy-on 1)}
        ]
            \pic
        \end{scope}
    \end{scope}
    \draw [spy] (spy-on node 1) -- (spy-in node 1);
\draw [gray!50!white,line width=1pt,fill=white] (0.7,2.3)rectangle (3,3.1);
\begin{scope}[shift={(0.6,2.4)}] 

	\draw [smooth,samples=1000,domain=0.0:2.2,red!60!black,mark=otimes,dash pattern=on 6pt off 6pt,line width=2pt] 
		{(0.25,0) --node [mark size=4pt,line width=0.1pt]{$\pgfuseplotmark{triangle*}$} (0.5,0)}
		node[right]{New achievable rate };
		\draw [yshift=\baselineskip,smooth,blue!70,samples=1000,domain=0.0:2.2,mark=otimes,dash pattern=on 10pt off 10pt,line width=2pt] 
				{(0.25,0) --node [mark size=2.5pt,line width=0.1pt]{$\pgfuseplotmark{square*}$} (0.5,0)}
				node[right]{Known achievable rate };
						\draw [yshift=2\baselineskip,smooth,gray!80!white,samples=1000,domain=0.0:2.2,mark=otimes,line width=2pt] 
								{(0.25,0) --node [mark size=3pt,line width=0.1pt]{$\pgfuseplotmark{*}$} (0.5,0)}
								node[right]{Lower bound for the rate \cite{tian2015note}};
	\end{scope}

\end{tikzpicture}
\captionsetup{justification=centering}
\caption{Rate memory tradeoff for the $(3,3)$ cache network.}
\label{fig:memory rate curve for (3,3) cache network}
\end{figure}
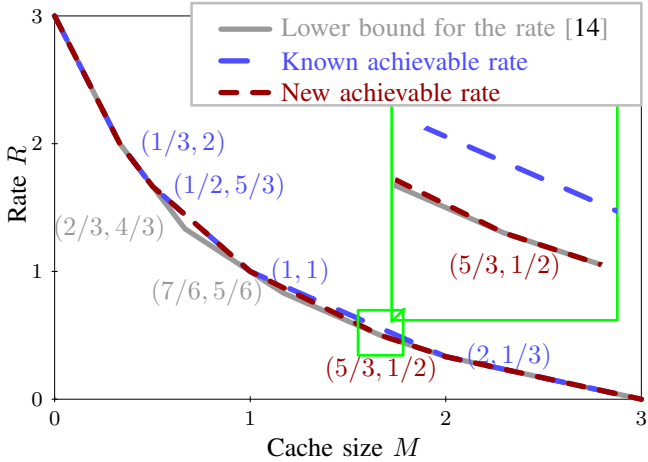

In \cite{tian2015note}, using the computational approach, Tian presented a lower bound for the $(3,3)$ cache network:
\begin{lemma}\label{(3,3) constraint} (Tian\label{3by3}\cite{tian2015note})
\label{3by3 bound}
For the $(3,3)$ cache network  achievable memory rate pairs $(M,R)$ must satisfy
\begin{IEEEeqnarray*}{c}
3M +6 R \geq 8.\label{c.5}
\end{IEEEeqnarray*}
\end{lemma}
For the $(3,3)$ cache network we can achieve the memory rate pair $(5/3, 1/2)$ by using the scheme presented above and the memory rate pair $(2,1/3)$ by using the scheme presented in \cite{maddah2014fundamental}. By memory sharing these schemes we can achieve memory rate pairs satisfying, 
\begin{equation}
R(M)=\frac{4}{3}-\frac{1}{2}M,
\end{equation}
where $M\in [5/3,2]$. From lemma \ref{3by3}, we have,
\begin{IEEEeqnarray}{rll}
R\geq&\  \frac{8}{6}-\frac{3}{6}M=\ R(M). 
\end{IEEEeqnarray}
This completes the characterization of the exact rate memory tradeoff when $M\in[5/3,2]$, for demands where users request different files. The best achievable memory rate pairs and the best lower bounds for the $(3,3)$ cache network is shown in Fig. \ref{fig:memory rate curve for (3,3) cache network}. An extension to the scheme described in this section is presented next.
\section{A New Coding Scheme}
Consider the $(K,K)$ cache network, where each user has a cache memory of size $M=(K-1-1/K)$. In this network, the server has $K$ files $\{W_{1}\dots,W_{K}\}$, each of size $F$ bits and $K$ users are connected to the server through a common shared link. In this section, we propose an extension to the scheme presented in section II. For brevity, we use the notation $[k]$ to represent the set $\{1,2,\dots,k\}$ and $W_{[k]}$ to represent the set $\{W_{1},\dots,W_{k}\}$.

During the placement phase, each file is split into $K(K-1)$ subfiles, each of size $F/(K(K-1))$ bits. Let the subfiles of the file $W_{n}$ be denoted by $W_{n,\{i,j\}}^{i}$ and $W_{n,\{i,j\}}^{j}$, for all $i,j\in[K]$ and $i\neq j$. The placement phase can be split into two stages. In first stage of the placement phase, the server places the subfiles $W_{n,\{i,j\}}^{i}$ and $W_{n,\{i,j\}}^{j}$, for all $n\in [K]$ into all users cache except the users $i$ and $j$. After the first stage of the placement phase, the user $k$'s cache contents are:
\begin{equation*}
\bigcup_{n\in [K]} \ \bigcup_{i,j\in{K\setminus\{k\}}, i\neq j}\Big( W_{n,\{i,j\}}^{i},W_{n,\{i,j\}}^{j}\Big),
\end{equation*}
Thus, in this stage of the placement phase, from each file the server places $2\times\binom{K-1}{2}$ subfiles are placed in each user's cache. The total space occupied by these subfiles in each users cache is $2\times K \times \frac{F}{K(K-1)}= F(K-1)\text{ bits.}$

In this paper we use the notation $W_{n,\{0,1\}}^{1}$ also to represent the subfile $W_{n,\{1,K\}}^{1}$. In the second stage of the placement phase, the server places functions of the subfiles in each user's cache. For the user $k$, the server places the functions $W_{n,\{k-1,k\}}^{k}\oplus W_{n,\{j,k\}}^{k}$, for all $n\in [K]$ and $j\in [K]\setminus\{k-1,k\}$. For each file there exists $(K-2)$ such functions, each of size $F/(K(K-1))$ bits requiring a space of $K(K-2)\times\frac{F}{K(K-1)}=F(\frac{K-2}{K-1}) \text{ bits}$ in each user's cache. In  this stage, the server also places the function $\oplus_{n\in [K]}W_{n,\{k-1,k\}}^{k}$, of size $F/(K(K-1))$ bits, into the user $k$'s cache. At the completion of the second stage of the placement phase, the user $k$'s cache contains the following functions as well:
\begin{IEEEeqnarray*}{ccc}
\bigcup_{n\in [K]}\bigcup_{j\in[K]\setminus\{k-1,k\}}\bigg(W_{n,\{k-1,k\}}^{k}\oplus W_{n,\{j,k\}}^{k}\bigg), \\  \text{ and }
\bigoplus_{n\in [K]}W_{n,\{k-1,k\}}^{k}.
\end{IEEEeqnarray*}

The total space occupied in each user's cache can be computed as,
\begin{IEEEeqnarray*}{rll}
F(K-2)+F\frac{K-2}{K-1}+\frac{F}{K(K-1)}=&(K-1-\frac{1}{K})F\\=&MF \text{ bits.}
\end{IEEEeqnarray*}

During the delivery phase, each user communicates their demand with the server. Consider a demand $\textbf{d}=\{W_{d_{1}}\dots,W_{d_{K}}\}$, where $W_{d_{k}}$ represents the file requested by the user ${k}$ and each user requests a different file. Consider the set $S_{i}=[K]\setminus\{i\}$, for $i\in[K]$. For each set $S_{i}$ the server broadcasts the packet $\oplus_{k\in S_{i}}W^{i}_{d_{k},\{k,i\}}$. The set of packets broadcasts in the delivery phase is,
\begin{equation*}
X_{\textbf{d}}=\bigcup_{i\in [K]}\bigoplus_{k\in S_{i}}W^{i}_{d_{k},\{k,i\}},
\end{equation*}
Since each packet is of size $F/(K(K-1))$ bits, the load experienced by the shared link is,
\begin{IEEEeqnarray*} {rll}
K\times \frac{F}{K(K-1)}=\frac{F}{K-1}=RF \text{ bits}.
\end{IEEEeqnarray*}
Thus, the scheme operates at the memory rate pair $(M,R)=\big(K-1-\frac{1}{K},\frac{1}{K-1}\big)$.

To understand how each user obtains its requested file, let us consider the user $k$, who has  requested the file $W_{d_{k}}$. This user has the subfiles $W_{d_{k},\{i,j\}}^{i}$ and $W_{d_{k},\{i,j\}}^{j}$, where $i,j\in[K]\setminus\{k\}$. This user needs the subfiles $W_{d_{k},\{k,i\}}^{i}$, and $W_{d_{k},\{k,i\}}^{k}$, where $i\in [K]\setminus\{k\}$, to compute  $W_{d_{k}}$. This user obtains these subfiles in two stages. In the first stage, this user obtains the subfiles $W_{d_{k},\{k,i\}}^{i}$. For a set $S_{i}$ the server broadcasts the packet,
\begin{equation}
\oplus_{k\in S_{i}}W^{i}_{d_{k},\{k,i\}}.
\end{equation}
It can be noted that the subfile $W_{d_{j},\{j,i\}}^{i}$, for all $j\in S_{i}$, are available to the user $k$. By using these subfiles the user ${k}$ obtains the subfile $W^{i}_{d_{k},\{k,i\}}$. This is true for all set $S_{i}$, where $i\in [K]\setminus\{k\}$.

Consider the  set of users $S_{k}=[K]\setminus\{k\}$. Corresponding to this set the server broadcasts the packets $\oplus_{j\in S_{k}}W^{k}_{d_{j},\{k,j\}}$. In the second, the cached contents $W^{k}_{d_{j},\{k-1,k\}}\oplus W^{k}_{d_{j},\{k,j\}}$, for all $j\in [K]\setminus\{k\}$, along with the received packet $W^{k}_{d_{j},\{k-1,k\}}\oplus W^{k}_{d_{j},\{k,j\}}$ can be used to compute the function,

\begin{equation}
\label{stage2eq1}
\oplus_{j\in S_{k}}W^{k}_{d_{j},\{k-1,k\}},
\end{equation}
This function along with the cached content,
\begin{equation*}
 \oplus_{n\in [K]}W_{n,\{k-1,k\}}^{k},
\end{equation*}
can be used to recover the subfile $W_{d_{k},\{k-1,k\}}^{k}$. This subfile along with the cached content,
\begin{equation*}
W^{k}_{d_{k},\{k-1,k\}} \oplus W^{k}_{d_{k},\{k,j\}},
\end{equation*}
can recover the subfile $W^{k}_{d_{k},\{k,j\}}$, for all $j\in[K]\setminus\{K\}$. Thus, user $k$ obtains all the subfiles of its requested file $W_{d_{k}}$. Similarly, other users can also obtain their requested files. Thus, we have,

\begin{theorem}
For the $(K,K)$  cache network  the memory rate pair $(K-1-\frac{1}{K},\frac{1}{K-1})$ is achievable for the demand where all users request for different files.

\end{theorem}
\begin{figure*}[!t]
\hrulefill
\begin{IEEEeqnarray*}{rll}
\overset{}{=}& H(W_{[K-1]},Z_{K},X^{[K]\setminus\{1\}})+H(W_{[K-1]},Z_{K-1},X^{[K]\setminus\{2\}})+\sum\limits_{i=3}^{K} H(W_{[K-1]},Z_{K+1-i},X^{[K]\setminus\{i\}})\\
\overset{(a)}{\geq}& H(W_{[K-1]},Z_{K},Z_{K-1},X^{[K]\setminus\{1\}},X^{[K]\setminus\{2\}})+H(W_{[K-1]},X^{[K]\setminus\{1,2\}})+\sum\limits_{i=3}^{K} H(W_{[K-1]},Z_{K+1-i},X^{[K]\setminus\{i\}})\\
\overset{(c)}{=}& H(W_{[  K-1]},W_{K},Z_{K},Z_{K-1},X^{[K]})+H(W_{[K-1]},X^{[K]\setminus[2]})+\sum\limits_{i=3}^{K} H(W_{[K-1]},Z_{K+1-i},X^{[  K]\setminus\{i\}})\\
\overset{(d)}{=}& H(W_{[K]})+H(W_{[K-1]},X^{[K]\setminus[2]})+H(W_{[K-1]},Z_{K-2},X^{[K]\setminus\{3\}})+\sum\limits_{i=4}^{K} H(W_{[K-1]},Z_{K+1-i},X^{[K]\setminus\{i\}})\\
\overset{(e)}{\geq}& H(W_{[K]})+H(W_{[K]})+H(W_{[K-1]},X^{[K]\setminus[3]})+H(W_{[K-1]},Z_{K-3},X^{[K]\setminus\{4\}}) +\sum\limits_{i=5}^{K} H(W_{[K-1]},Z_{K+1-i},X^{[K]\setminus\{i\}}) \\
\overset{(e)}{\geq}& 2H(W_{[K]})+H(W_{[K]})+H(W_{[K-1]},X^{[K]\setminus[4]})+H(W_{[K-1]},Z_{K-4},X^{[K]\setminus\{5\}}) +\sum\limits_{i=6}^{K} H(W_{[K-1]},Z_{K+1-i},X^{[K]\setminus\{i\}}) 
\end{IEEEeqnarray*}
\hrulefill
\end{figure*}

\section{A New Lower Bound}
In this section we prove the optimality of the scheme presented in the previous section for the $(K,K)$  cache network. We use the following identities.
\begin{IEEEeqnarray}{rll}
H(Z_{k},X_{\textbf{d}})=&H(W_{d_{k}},Z_{k},X_{\textbf{d}}),\label{I.1}\\
H(W_{[K]},Z_{k},X_{\textbf{d}})=&H(W_{[K]}),\label{I.2}
\end{IEEEeqnarray}
where $(\ref{I.1})$ follows from the fact that each user $k$ can recover the requested file, $W_{d_{k}}$, from the received packets, $X_{\textbf{d}}$, with the help of its cached contents, $Z_{k}$, and $(\ref{I.2})$ follows from the fact that each user's cached contents, $Z_{k}$, and the packets broadcast in response to the demand $\textbf{d}$, $X_{\textbf{d}}$, are functions of files in the server, $W_{[K]}$. Consider a demand $\textbf{d}^{i}=$ $\{W_{i},W_{i+1},\dots,W_{K},$ $W_{1},\dots,W_{i-1}\}$, for $i \in [K]$. In response to this demand $\textbf{d}^{i}$, the server broadcast a set of packets $X^{i}$. We use the notation $X^{[k]}$ to represent the set $\{X^{1},\dots, X^{k}\}$. It can be noted that, the user $U_{K+1-i}$ requests for all the files except the file $W_{K}$ in the demands $\textbf{d}^{[k]\setminus\{i\}}$ and  in the demand $\textbf{d}^{i}$, the user $U_{K+1-i}$ requests for the file $W_{K}$. Thus, from (\ref{I.1}) we have,
\begin{IEEEeqnarray}{rll}
\label{requestall}
H(Z_{K+1-i},X^{[K]\setminus\{i\}})=&H(W_{[K-1]},Z_{K+1-i},X^{[K]\setminus\{i\}})\text{,  \ \ \ \ } \\
\label{request}
H(Z_{K+1-i},X^{i})=&H(W_{K},Z_{K+1-i},X^{i}).
\end{IEEEeqnarray}
The following lemma can now be established.
\begin{lemma}
\label{reduction}
For the $(K, K)$  cache network we have the identity:
\begin{IEEEeqnarray*}{rll}
H(W_{[K-1]},X^{[K]\setminus S})+&H(W_{[K-1]},Z_{K+1-i},X^{[K]\setminus \{i\}})\\\geq &H(W_{[K-1]},X^{[K]\setminus S\cup\{i\}})+H(W_{[K]}),
\end{IEEEeqnarray*}
where $S\subset[K]$ and $i\notin S$.
\end{lemma}

\begin{proof} We have,
\begin{IEEEeqnarray*}{rll}
H(W&_{[K-1]},X^{[K]\setminus S})+H(W_{[ K-1]},Z_{K+1-i},X^{[K]\setminus \{i\}})\\
=& H(W_{[K-1]},X^{[K]\setminus S\cup \{i\}},X^{i})\\&+H(W_{[ K-1]},Z_{K+1-i},X^{[ K]\setminus S\cup \{i\}},X^{S})\\
\overset{(a)}{\geq}& H(W_{[ K-1]},X^{[K]\setminus S\cup \{i\}})+H(W_{[ K-1]},Z_{K+1-i},X^{[K]})\\
\overset{(b)}{=}& H(W_{[K-1]},X^{[K]\setminus S\cup \{i\}})+H(W_{[K]},Z_{K+1-i},X^{[K]})\\
\overset{(c)}{=}& H(W_{[K-1]},X^{[K]\setminus S\cup \{i\}})+H(W_{[K]}),
\end{IEEEeqnarray*}
where $(a)$ follows from submodularity of entropy, $(b)$ follows from (\ref{request}) and $(c)$ follows from (\ref{I.2}).
\end{proof}
Using this lemma repeatedly, we prove the following theorem.
\begin{theorem}
\label{neqk}
For the $(K,K)$  cache network the achievable memory rate pair must satisfy the following constraint.
\begin{equation*}
KM+K(K-1)R\geq K^2-1.
\end{equation*}
\end{theorem}
\begin{proof}
We have,
\begin{IEEEeqnarray*}{rll}
KM+K(K-&1)R \geq \sum\limits_{i=1}^{K}H(Z_{K+1-i})+\sum\limits_{i=1}^{K}H(X^{[K]\setminus\{i\}})\\
\overset{(a)}{\geq}&\sum\limits_{i=1}^{K} H(Z_{K+1-i},X^{[K]\setminus\{i\}})\\
\overset{(b)}{=}&\sum\limits_{i=1}^{K} H(W_{[K-1]},Z_{K+1-i},X^{[K]\setminus\{i\}})\\
=&(\text{See top of next page}) \\
\overset{(f)}{\geq}& (K-2)H(W_{[K]})+H(W_{[K-1]},X^{[K]\setminus[K-1]}) \\&+H(W_{[K-1]},Z_{1},X^{[K]\setminus\{K\}}) \\
\overset{(e)}{\geq}& (K-1)H(W_{[K]})+H(W_{[K-1]}) \\
\overset{}{=}& (K-1)K+K-1=K^{2}-1,
\end{IEEEeqnarray*}
where $(a)$ follows from submodularity of entropy, $(b)$  follows from (\ref{requestall}), $(c)$ follows from (\ref{request}), $(d)$ follows from (\ref{I.2}), $(e)$ follows from lemma \ref{reduction} and $(f)$ follows from applying lemma \ref{reduction} for $K-5$ times.
\end{proof}

For the $(K,K)$  cache network, the proposed scheme achieves the memory rate pair $(K-1-1/K,1/K-1)$ and the scheme presented in \cite{maddah2014fundamental} achieves the memory rate pair $((K-1),1/K)$. By memory sharing these schemes, the rate achieved when $M\in [K-1-\frac{1}{K},K-1]$ is,
\begin{equation*}
R(M)=\frac{K^{2}-1}{K(K-1)}-\frac{M}{K-1}.
\end{equation*}
From theorem \ref{neqk}, we have,
\begin{align*}
R\geq \frac{K^{2}-1}{K(K-1)}-\frac{M}{K-1}=\ R(M).
\end{align*}
This completes the characterization of the exact rate memory tradeoff for cache size  $M\in[\frac{K^{2}-K-1}{K},K-1]$. The results of sections III and IV are summarised in the following theorem:
\begin{theorem}
For the $(K,K)$ cache network, when $M\in [\frac{K^{2}-K-1}{K},K-1]$, the exact rate memory treadeoff is characterised by,
\begin{equation*}
R(M)=\frac{K+1}{K}-\frac{M}{K-1}.
\end{equation*}
\end{theorem}
\section{Conclusions}
In this paper we considered the $(K,K)$ cache network and proposed a new coding scheme to achieve the memory rate pair $(\frac{K^{2}-K-1}{K},\frac{1}{K-1})$, when all users are requesting for different files. Further, we proved the optimality of the proposed scheme. We are currently pursuing the problems of extending this scheme to the general $(N,K)$ cache network where each user has a cache memory of size $M\in[0,N]$ and obtain a matching lower bound for this network.

\ifCLASSOPTIONcaptionsoff
  \newpage
\fi

\end{document}